\documentclass[11pt,onecolumn,draftcls]{IEEEtran}
\usepackage{amssymb}
\usepackage{amsmath}
\usepackage{amsfonts}
\usepackage{float}
\usepackage{graphics}
\usepackage[tight,footnotesize]{subfigure}
\usepackage{rotating}
\usepackage{algorithm}
\usepackage{algorithmic}
\usepackage{enumerate}

\newcommand*{\QEDA}{\hfill\ensuremath{\blacksquare}}

\DeclareMathSizes{10}{9}{7}{5}
\newtheorem{theo}{Theorem}

\newtheorem{theorem}{Theorem}

\newtheorem{definition}[theorem]{Definition}

\newtheorem{remark}{Remark}

\input{tcilatex}
\begin{document}

\title{Practical Secrecy using Artificial Noise}
\pubid{}
\specialpapernotice{}
\author{Shuiyin~Liu,~Yi~Hong,~and~Emanuele~Viterbo\thanks{%
S. Liu, Y. Hong and E. Viterbo are with the Department of Electrical and
Computer Systems Engineering, Monash University, Melbourne VIC 3180,
Australia (e-mail: shuiyin.liu, yi.hong, emanuele.viterbo@monash.edu). This
work was performed at the Monash Software Defined Telecommunications Lab and
the authors were supported by the Monash Professional Fellowship, 2013
Monash Faculty of Engineering Seed Funding Scheme, and Australian Research
Council Discovery grants (ARC DP130100336).}}

%\author{Shuiyin~Liu,~Yi~Hong,~and~Emanuele~Viterbo\thanks{%
%S. Liu, Y. Hong and E. Viterbo are with the Department of Electrical and
%Computer Systems Engineering, Monash University, Melbourne VIC 3180,
%Australia (e-mail: shuiyin.liu@monash.edu, yi.winnie.hong@gmail.com and
%emanuele.viterbo@monash.edu).}}

%DS

%%EndExpansion
%TCIMACRO{\TeXButton{Make Title}{\maketitle}}%
%BeginExpansion
\maketitle%
%EndExpansion

%TCIMACRO{\TeXButton{Begin abstract}{\begin{abstract}}}%
%BeginExpansion
\begin{abstract}%
%EndExpansion

In this paper, we consider the use of artificial noise for secure
communications. We propose the notion of \emph{practical secrecy} as a new
design criterion based on the behavior of the eavesdropper's error
probability $P_{\text{E}}$, as the signal-to-noise ratio goes to infinity.
We then show that the \emph{practical secrecy }can be guaranteed by the
randomly distributed artificial noise with specified power. We show that it
is possible to achieve \emph{practical secrecy} even when the eavesdropper
can afford more antennas than the transmitter.

%TCIMACRO{\TeXButton{End abstract}{\end{abstract}}}%
%BeginExpansion
\end{abstract}%
%EndExpansion
\vspace{-3mm}

\section{INTRODUCTION}

The \textit{broadcast }characteristic of wireless communication systems
results in enormous challenges in securing transmitted data in the presence
of eavesdroppers. The eavesdropper is commonly assumed to be passive and its
location is unknown to the transmitter. In current wireless systems, secure
communication mainly depends on the network layer cryptographic
technologies. Information theoretic results show that it is possible to
secure the data by employing physical layer strategy \cite{Wyner75}, when
the intended receiver has a better channel than the eavesdropper. The
secrecy capacity is thus defined to measure the difference between the
capacities of the intended user and the eavesdropper \cite{Hellman78}.

A recently proposed physical layer security scheme makes use of artificial
noise to degrade the eavesdropper's reception \cite{Goel08}. The intended
user is unaffected, so that a non-zero secrecy rate is ensured. The approach
assumes a Gaussian artificial noise and requires that the number of
eavesdropper antennas $N_{\text{E}}$ is strictly smaller than the number of
transmitter antennas $N_{\text{A}}$. In this paper we tackle this problem.
To overcome the restriction $N_{\text{E}}<N_{\text{A}}$, we aim at
maximizing the eavesdropper's error probability, defined by $P_{\text{E}}$,
rather than the secrecy rate. Hence, we define the notion of \emph{practical
secrecy}\ as $P_{\text{E}}\rightarrow 1$ exponentially as the number of
receiver antennas $N_{\text{B}}\rightarrow \infty $, for any signal-to-noise
ratio (SNR) at the eavesdropper. The proposed criterion is different from
the \emph{secrecy gain} introduced in \cite{oggerb11c}, where $P_{\text{E}%
}\rightarrow 0$ for high eavesdropper SNR. More importantly, we
propose the \emph{covering ratio} as a fundamental secrecy parameter
which guarantees the convergence of $P_{\text{E}}$ and characterizes
the amount of the artificial noise required. Furthermore, we propose
lattice precoding to improve performance over singular value
decomposition (SVD) precoding used in \cite{Goel08}.

The paper is organized as follows: Section II presents the transmission
model and lattice basics. The SVD precoding and the lattice precoding are
given in Section III. In Section IV, \emph{practical secrecy} is addressed.
Section V presents some simulation results. Some concluding remarks are
drawn in Section VI. Proofs of the theorems are given in Appendix.

\textit{Notation:} Matrices and column vectors are denoted by upper and
lowercase boldface letters, and the transpose, inverse, pseudoinverse of a
matrix $\mathbf{B}$ by $\mathbf{B}^{T}$, $\mathbf{B}^{-1}$, and $\mathbf{B}%
^{\dagger }$, respectively. $X\rightarrow Y$ denotes that random variable $X$
converges to random variable $Y$ in distribution. We use the standard
asymptotic notation $f\left( x\right) =O\left( g\left( x\right) \right) $,
when $\lim \sup_{x\rightarrow \infty }|f(x)/g(x)|<\infty $. $\mathbb{R}$, $%
\mathbb{C}$, $\mathbb{Z}$ and $\mathbb{Z}\left[ i\right] $ represent the
real, complex, integer and complex integer numbers, respectively. \vspace{%
-3mm}

\section{System Model and Lattice Preliminary}

We consider the multiple-input multiple-output (MIMO) wiretap channel using
the \textit{M}-QAM signalling. The precoding and decoding problems in this
system can be easily modeled using lattices \cite{Hochwald05}\cite{mow:IT}.
In what follows, the system model is introduced first, followed by some
lattice preliminaries that are relevant to this paper. \vspace{-3mm}

\subsection{System Model}

Consider a MIMO wiretap system including three terminals:\ a transmitter
Alice, an intended receiver Bob, and a passive eavesdropper Eve, which are
equipped with $N_{\text{A}}$, $N_{\text{B}}$ and $N_{\text{E}}$ antennas,
respectively. Bob and Eve receive%
\begin{equation}
\mathbf{z}\mathbf{=Hx+n}_{\text{B}},  \label{secrecy_signal}
\end{equation}%
\begin{equation}
\mathbf{y}\mathbf{=Gx}+\mathbf{n}_{\text{E}},  \label{Eve_signal}
\end{equation}%
respectively, where $\mathbf{n}_{\text{B}}\in \mathbb{C}^{N_{\text{B}}\times
1}$ and $\mathbf{n}_{\text{E}}\in \mathbb{C}^{N_{\text{E}}\times 1}$ are the
complex white Gaussian noise vectors with i.i.d. entries $\sim \mathcal{N}_{%
\mathbb{C}}(0$, $\sigma _{\text{B}}^{2})$ and $\mathcal{N}_{\mathbb{C}}(0$, $%
\sigma _{\text{E}}^{2})$, respectively. Assuming Bob and Eve are not
co-located, then the mutually independent matrices $\mathbf{H}\in \mathbb{C}%
^{N_{\text{B}}\times N_{\text{A}}}$ and $\mathbf{G}\in \mathbb{C}^{N_{\text{E%
}}\times N_{\text{A}}}$ represent the channels from Alice to Bob and Alice
to Eve, respectively, where the entries are assumed to be i.i.d. circularly
symmetric Gaussian random variable $\sim \mathcal{N}_{\mathbb{C}}(0$, $1)$.

We assume $N_{\text{B}}<N_{\text{A}}$, so that $\mathbf{H}$ has a
non-trivial null space generated by the columns of the matrix $\mathbf{Z}=$%
\;null($\mathbf{H}$). Let $\mathbf{u}$ be the secret data vector. Using the
artificial noise technique, Alice sends%
\begin{equation}
\mathbf{x=Pu+Zv,}  \label{T1}
\end{equation}%
where $\mathbf{P}$ is the precoding matrix and $\mathbf{v}$ is the
artificial noise generated by Alice. Considering uniform $M$-QAM
signalling, we have the secret data $\Re (\mathbf{u)}$ and $\Im (\mathbf{u)}%
\in \mathcal{C}^{N_{\text{B}}}$, where $\mathcal{C=}\{-\sqrt{M}+1$, $-\sqrt{M%
}+3$ \ldots,$\sqrt{M}-1\}$. The total transmit power is constrained to $P$,
i.e., E$[||\mathbf{x}||^{2}\mathbf{]}\leq P$.

Then, (\ref{secrecy_signal}) and (\ref{Eve_signal}) can be rewritten as%
\begin{align}
\mathbf{z}& \mathbf{=HPu+n}_{\text{B}}  \label{sec_mod2} \\
\mathbf{y}& \mathbf{=GPu+GZv+n}_{\text{E}}\text{\text{.}}  \label{Eve_mod2}
\end{align}

We consider the worse case for Alice where Eve not only knows the channel
matrices $\mathbf{H}$ and $\mathbf{G}$, but also knows the matrix $\mathbf{Z}
$ and the precoding matrix $\mathbf{P}$. Alice is assumed to know only $%
\mathbf{H}$. The SNR of Eve is defined as SNR$_{\text{E}}\triangleq P/\sigma
_{\text{E}}^{2}$. From (\ref{sec_mod2}) and (\ref{Eve_mod2}), through the
interference term $\mathbf{GZv}$, we can see that $\mathbf{v}$ affects Eve,
but not Bob. \vspace{-3mm}

\subsection{Lattice Preliminary}

An $n$-dimensional real \emph{lattice} in an $m$-dimensional Euclidean space
$\mathbb{R}^{m}$ ($n\leq m$) is the set of integer linear combinations of $n$
independent vectors:%
\begin{equation*}
\Lambda _{\mathbb{R}}=\left\{ \mathbf{Bu}\text{ : }\mathbf{u\in }\text{ }%
\mathbb{Z}^{n}\right\} \text{,}
\end{equation*}%
where $\mathbf{B=}\left[ \mathbf{b}_{1}\cdots \mathbf{b}_{n}\right] $ is a
\emph{basis} of the lattice $\Lambda _{\mathbb{R}}$.

In the following, we introduce some lattice parameters which are related to
this work.

A \emph{shortest vector} of $\Lambda _{\mathbb{R}}$ is a non-zero vector in $%
\Lambda _{\mathbb{R}}$ with the smallest Euclidean norm. The length of the
shortest vector is denoted by $\lambda _{1}\left( \mathbf{B}\right) $.

The \emph{Voronoi region} of a lattice point $\mathbf{x}_{i}$ is denoted by:%
\begin{equation*}
\mathcal{V}\left( \Lambda _{\mathbb{R}}\right) =\left\{ \mathbf{y}\in
\mathbb{R}^{m}\text{: }\Vert \mathbf{y}-\mathbf{x}_{i}\Vert \leq \Vert
\mathbf{y}-\mathbf{x}_{j}\Vert ,\forall \text{ }\mathbf{x}_{i}\neq \mathbf{x}%
_{j}\right\} .
\end{equation*}

The determinant of $\Lambda _{\mathbb{R}}$, $\det (\Lambda _{\mathbb{R}%
})\triangleq \sqrt{\det (\mathbf{B}^{T}\mathbf{B})}$, gives the $n$%
-dimensional volume of $\mathcal{V}\left( \Lambda _{\mathbb{R}}\right) $.

In Fig. 1, we illustrate two important lattice parameters which are related
to $\mathcal{V}\left( \Lambda _{\mathbb{R}}\right) $:

\begin{enumerate}
\item the \emph{effective radius} of $\Lambda _{\mathbb{R}}$, denoted by $r_{%
\text{eff}}(\Lambda _{\mathbb{R}})$, is the radius of a sphere $\mathcal{S}_{%
\text{eff}}(\Lambda _{\mathbb{R}})\,$of volume $\det (\Lambda _{\mathbb{R}})$
\cite{Zamir08}. For large $n$, it is approximately%
\begin{equation*}
r_{\text{eff}}(\Lambda _{\mathbb{R}})\approx \sqrt{n/(2\pi e)}\det (\Lambda
_{\mathbb{R}})^{1/n}\text{;}
\end{equation*}

\item the \emph{covering radius} of $\Lambda _{\mathbb{R}}$, denoted by $r_{%
\text{cov}}(\Lambda _{\mathbb{R}})$, is the radius of the smallest sphere
centred at a lattice point which covers $\mathcal{V}\left( \Lambda _{\mathbb{%
R}}\right) $.
\end{enumerate}

In wireless communication, it is common to use complex number representation
of signals. The real lattice definition can be extended to complex:%
\begin{equation*}
\Lambda _{\mathbb{C}}=\left\{ \mathbf{B}_{\mathbb{C}}\mathbf{u}_{\mathbb{C}}%
\text{ : }\mathbf{u}_{\mathbb{C}}\mathbf{\in }\text{ }\mathbb{Z}\left[ i%
\right] ^{n}\right\} \text{,}
\end{equation*}%
where $\mathbf{B}_{\mathbb{C}}\in \mathbb{C}^{m\times n}$ is a\emph{\ basis}
of the \emph{complex lattice} $\Lambda _{\mathbb{C}}$. There is a simple way
to represent $n$-dimensional complex lattices as $2n$-dimensional real
lattices \cite{BK:Conway93}. In this work, when we use the lattice
parameters of $\Lambda _{\mathbb{C}}$ (e.g., $\mathcal{V}\left( \Lambda _{%
\mathbb{C}}\right) $), we first convert $\Lambda _{\mathbb{C}}$ to the real
equivalent $\Lambda _{\mathbb{R}}$, and then apply the corresponding
definitions of $\Lambda _{\mathbb{R}}$.

From the lattice viewpoint, $\mathbf{GPu}$ in (\ref{Eve_mod2})\textbf{\ }can
be described as a point of the lattice with a basis $\mathbf{GP}$. The
detection of $\mathbf{u}$ fits in the lattice decoding scenario and can be
solved by sphere decoding \cite{viterbo}. In this paper, we assume the
worst-case for Alice and Bob, where Eve is able to perform maximum
likelihood decoding (e.g., by sphere decoding) to estimate $\mathbf{u}$,
even if the average complexity grows exponentially with the lattice
dimension.
\begin{figure}[tbp]
\centering
\par
\includegraphics[scale=0.4]{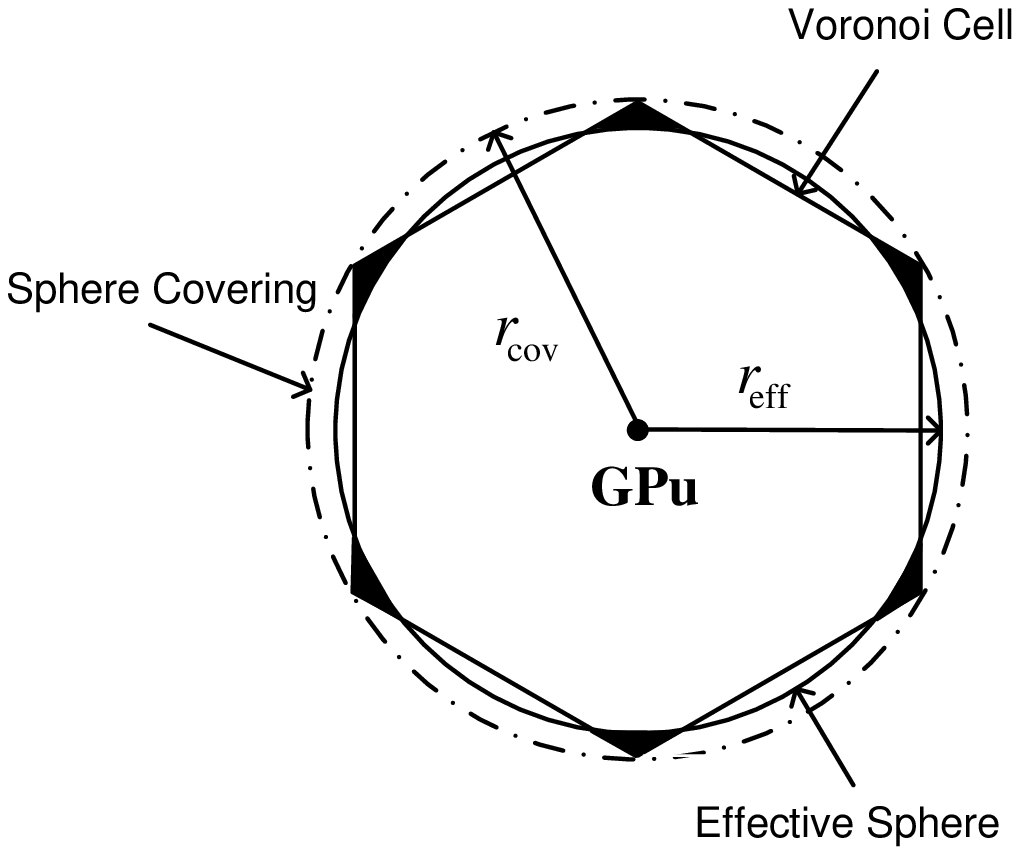} \vspace{-2mm}
\par
\caption{Voronoi Cell, Effective Radius and Covering
Radius.}\vspace{-5mm}
\end{figure}

\section{Precoding for Secure Communication}

In this Section, we analyze two different precoding schemes for the
artificial noise strategy: SVD precoding and lattice precoding.

For the MIMO scenario, the original artificial noise stragegy\emph{\ }\cite%
{Goel08} uses SVD precoding, where $\mathbf{H=U\Lambda V}^{T}$, $\mathbf{V=[V%
}_{1}$, $\mathbf{Z]}$ and $\mathbf{P=V}_{1}$. Due to the orthogonality
between $\mathbf{P}$ and $\mathbf{Z}$, from (\ref{T1}), the total
transmission power is%
\begin{equation}
||\mathbf{x}_{\text{SVD}}||^{2}=||\mathbf{u}||^{2}+||\mathbf{v}_{\text{SVD}%
}||^{2}.  \label{SVD_POW}
\end{equation}

Different from SVD precoding, lattice precoding \cite{Hochwald05} transmits%
\begin{equation}
\mathbf{x}_{\text{LP}}\mathbf{={\mathbf{H}^{\dag }{(\mathbf{u}-}}}A\mathbf{{{%
\mathbf{\hat{w}})}}+Zv}\text{\text{,}}
\end{equation}%
where $A=2\sqrt{M}$ and%
\begin{equation}
\mathbf{\hat{w}}=\arg \min_{\mathbf{w}\in \mathbb{Z}\left[ i\right] ^{N_{%
\text{B}}}}{\Vert \mathbf{H}^{\dag }{(\mathbf{u}-A\mathbf{w})}\Vert ^{2}}%
\text{.}  \label{min_power}
\end{equation}

Hence the corresponding transmission power is%
\begin{equation}
||\mathbf{x}_{\text{LP}}||^{2}=||{\mathbf{H}^{\dag }{(\mathbf{u}-A\mathbf{{%
\mathbf{\hat{w}}}})}}||^{2}+||\mathbf{v}_{\text{LP}}||^{2}\text{.}
\label{LP_POW}
\end{equation}

The search in (\ref{min_power}) requires the use of sphere decoder. To speed
up the search process, we apply the lattice reduction aided successive
interference cancellation (LR-SIC) precoding \cite{Windpassinger04}, in
which $\mathbf{\hat{w}}$ is approximated by Babai's nearest plane algorithm
\cite{Babai}.

In the next Section, we will show that lattice precoding outperforms SVD
precoding by requiring lower artificial noise power. \vspace{-5mm}

\section{Practical Secrecy}

In this Section, we propose a new artificial noise strategy to overcome the
limitation of $N_{\text{E}}<N_{\text{A}}$ and the assumption of Gaussian
artificial noise in \cite{Goel08}. Instead of targeting a non-zero secrecy
rate, the proposed scheme aims to maximize Eve's error probability $P_{\text{%
E}}\triangleq \Pr (\mathbf{\hat{u}}_{\text{E}}\neq \mathbf{u)}$. \vspace{-3mm%
}

\subsection{Practical Secrecy}

Let $\mathbf{\hat{u}}_{\text{E}}$ be the estimated secret message at Eve. We
propose a new measure of secrecy in terms of $P_{\text{E}}$.

\begin{definition}
We say \emph{practical secrecy} is achieved if for any SNR$_{\text{E}}$, $P_{%
\text{E}}\rightarrow 1$ exponentially as $N_{\text{B}}\rightarrow \infty $.
\end{definition}

The traditional secrecy capacity criterion, is based on the assumption of a
Gaussian input alphabet. On the contrary, \emph{practical secrecy} is
proposed for the practical communication systems, which make use of a finite
alphabet (e.g., $M$-QAM).

To the best of the authors' knowledge, no scheme has been proposed in the
literature to achieve \emph{practical secrecy}. In the following, we will
evaluate the relationship between \emph{practical secrecy} and $\mathbf{v}$
under the assumption that Eve can perform maximum likelihood decoding.%
\vspace{-3mm}

\subsection{Achieving Practical Secrecy}

Consider the lattice $\Lambda _{\mathbb{C}}$ with a basis $\mathbf{GP}$. The
decoding region of the target lattice point $\mathbf{GPu}$ is its associated
$\mathcal{V}\left( \Lambda _{\mathbb{C}}\right) $. Therefore, $P_{\text{E}}$
is determined by whether $\mathbf{y}$ in (\ref{Eve_mod2}) belongs to $%
\mathcal{V}\left( \Lambda _{\mathbb{C}}\right) $ or not. Let $\mathbf{\tilde{%
n}}_{\text{E}}=\mathbf{GZv+n}_{\text{E}}$ be Eve's generalized noise term.
For a given $\mathbf{v}$, the entries of $\mathbf{\tilde{n}}_{\text{E}}$ are
i.i.d. random variables $\sim \mathcal{N}_{\mathbb{C}}(0$, $\tilde{\sigma}_{%
\text{E}}^{2})$ with $\tilde{\sigma}_{\text{E}}^{2}=||\mathbf{v}%
||^{2}+\sigma _{\text{E}}^{2}$. A salient feature is that Eve's channel
noise $\mathbf{n}_{\text{E}}$ can help Alice to save on the artificial noise
power. In this work, we consider the worst-case scenario, i.e., $\sigma _{%
\text{E}}^{2}\rightarrow 0$, so that $P_{\text{E}}$ only depends on $\mathbf{%
GZv}$ and is independent of SNR$_{\text{E}}$.

As shown in Fig. 1, if the interference term $||\mathbf{GZv||}\geq r_{\text{%
cov}}(\Lambda _{\mathbb{C}})$, then $\mathbf{y}\notin \mathcal{V}\left(
\Lambda _{\mathbb{C}}\right) $, so that $P_{\text{E}}=1$. If%
\begin{equation}
\frac{r_{\text{cov}}(\Lambda _{\mathbb{C}})}{r_{\text{eff}}(\Lambda _{%
\mathbb{C}}\mathbf{)}}\geq \frac{||\mathbf{GZv||}}{r_{\text{eff}}(\Lambda _{%
\mathbb{C}}\mathbf{)}}>1\text{,}  \label{Cr_idear}
\end{equation}%
there are two cases: $P_{\text{E}}=1$ when $\mathbf{y}\in \mathcal{\bar{S}}_{%
\text{eff}}(\Lambda _{\mathbb{C}})-\mathcal{V}\left( \Lambda _{\mathbb{C}%
}\right) $ and $P_{\text{E}}=0$ when $\mathbf{y}\in \mathcal{V}\left(
\Lambda _{\mathbb{C}}\right) -\mathcal{S}_{\text{eff}}(\Lambda _{\mathbb{C}%
}) $ (the shaded corners), where $\mathcal{\bar{S}}_{\text{eff}}(\Lambda _{%
\mathbb{C}})$ is the complement of $\mathcal{S}_{\text{eff}}(\Lambda _{%
\mathbb{C}})$. As $||\mathbf{GZv||}$ approaches $r_{\text{cov}}(\Lambda _{%
\mathbb{C}})$, the shaded corners will disappear. In other words, Eve has a
higher error floor as $||\mathbf{GZv||}$ increases from $r_{\text{eff}%
}(\Lambda _{\mathbb{C}}\mathbf{)}$ to $r_{\text{cov}}(\Lambda _{\mathbb{C}})$%
. Note that the idea is directly applicable to lattice precoding, where the
target lattice point $\mathbf{G{\mathbf{P}}u}$ is simply replaced by the
lattice point $\mathbf{G}{\mathbf{H}^{\dag }}\left( \mathbf{u-}A\mathbf{\hat{%
w}}\right) $. The secret data therefore becomes $\mathbf{u-}A\mathbf{\hat{w}}
$.

Inspired by (\ref{Cr_idear}), we now introduce a new secrecy parameter
related to $P_{\text{E}}$.

\begin{definition}
The \emph{covering ratio} is defined as%
\begin{equation}
c_{\text{R}}\triangleq \frac{||\mathbf{GZv||}}{r_{\text{eff}}(\Lambda _{%
\mathbb{C}}\mathbf{)}}\text{.}  \label{Cr_def}
\end{equation}
\end{definition}

The $c_{\text{R}}$ and $P_{\text{E}}$ are related by the following theorem.

\begin{theo}
\label{Th3}For $c_{\text{R}}\geq \pi e$ and $N_{\text{B}}\rightarrow \infty $%
, $P_{\text{E}}=1$ for any value of SNR$_{\text{E}}$.
\end{theo}

\begin{proof}
See Appendix A.
\end{proof}

To apply Theorem \ref{Th3}, we need to find the sufficient condition of $c_{\text{R}%
}\geq \pi e$. Since $c_{\text{R}}$ is a random variable depending on the
random channel matrix $\mathbf{G}$, the problem then reduces to finding the
sufficient condition on $\Pr \{c_{\text{R}}<\beta \}\rightarrow 0$ for some $%
\beta >0$.

\begin{theo}
\label{Th1}For $N_{\text{B}}\rightarrow \infty $, let $||\mathbf{v}||=\beta
e/\Phi $, where%
\begin{eqnarray}
\Phi _{\text{LP}} &=&\left[ \frac{(N_{\text{E}}-N_{\text{B}})!}{(N_{\text{A}%
}-N_{\text{B}})!}\cdot \frac{N_{\text{A}}!}{N_{\text{E}}!}\right] ^{\frac{1}{%
2N_{\text{B}}}} \;\;\;\;\;\;%
\begin{array}{c}
\text{{\small for lattice}} \\
\text{{\small precoding}}%
\end{array}%
\label{Pi_LP} \\
\Phi _{\text{SVD}} &=&\left[ \frac{(N_{\text{E}}-N_{\text{B}})!}{N_{\text{E}%
}!}\cdot N_{\text{B}}{}^{1/2}\right] ^{\frac{1}{2N_{\text{B}}}} \;\;\;\;%
\begin{array}{c}
\text{{\small for SVD}} \\
\text{{\small precoding}}%
\end{array}%
\label{Pi_SVD}
\end{eqnarray}%
then%
\begin{equation}
\Pr \{c_{\text{R}}<\beta \}\leq O\left( e^{-\min (N_{\text{B}}^{2}/\log (N_{%
\text{B}})\text{, }N_{\text{E}})}\right) .
\end{equation}
\end{theo}

\begin{proof}
See Appendix B.
\end{proof}

From Theorem \ref{Th3} with $\beta =\pi e$ and \ref{Th1}, the convergence behavior
of $\Pr \{c_{\text{R}}<\pi e\}$ implies $P_{\text{E}}\rightarrow 1$
exponentially as $N_{\text{B}}\rightarrow \infty $.

\begin{remark}
\label{Rm1}\emph{Practical secrecy }is achieved when $||\mathbf{v}||\geq \pi
e^{2}/\Phi $, where $\Phi $ is given in (\ref{Pi_LP}) and (\ref{Pi_SVD}),
depending on precoders.
\end{remark}

\begin{remark}
\label{Rm2}Since $\Phi _{\text{SVD}}<\Phi _{\text{LP}}$, $||\mathbf{v}_{%
\text{SVD}}||>||\mathbf{v}_{\text{LP}}||$.
\end{remark}

As shown above, \emph{practical secrecy} is only related to $||\mathbf{v}||$%
. However, if $\mathbf{v}$ is an integer vector, the term $\mathbf{\tilde{x}%
=GPu+GZv}$ in (\ref{Eve_mod2}) can be viewed as a lattice point of $\tilde{%
\Lambda}_{\mathbb{C}}$ with a basis $[\mathbf{GP}$, $\mathbf{GZ]}$, so that
Eve may be able to recover $\mathbf{\tilde{x}}$ by using sphere decoding. To
avoid this, we generate a continuous random vector $\mathbf{v}$, so that $%
\mathbf{\tilde{x}}$ can never be a lattice point of $\tilde{\Lambda}_{%
\mathbb{C}}$, hence can not be detected.

Since \emph{practical secrecy} requires that $P_{\text{E}}$
approaches 1 exponentially, it implies that even for small values of
$N_{\text{B}}$ the $P_{\text{E}}$ is very close to 1.
The simulation in the following section shows that our analysis is applicable
to a real system with finite numbers of antennas, even with $N_{\text{E}}>N_{\text{A}%
}$.\vspace{-2mm}

\section{Simulation Results}

This section examines the performance of the proposed artificial noise
scheme in the most favorable case for Eve, i.e., SNR$_{\text{E}}\rightarrow
\infty $. We construct $\mathbf{v}$ in two steps: 1) generating a vector
with $N_{\text{A}}-N_{\text{B}}$ uniformly distributed random variables; 2)
normalizing the length of the random vector to $\beta e/\Phi $.

Fig. 2 shows the error performances at Bob and Eve for an uncoded system
using $64$-QAM with $N_{\text{A}}=10$, $N_{\text{B}}=9$ and $N_{\text{E}}=20$%
. Reference \cite{Goel08} argued that the non-zero secrecy rate can not be
guaranteed when $N_{\text{A}}<N_{\text{E}}$. Nevertheless, our \emph{%
practical secrecy} criterion provides the opportunity to protect $\mathbf{u}
$ in these scenarios. With $\beta =1$, the result in Fig. 2 shows $P_{\text{E%
}}=1$, even when $N_{\text{A}}-N_{\text{B}}=1$. We find that $\beta =1$ is
already good enough in practice. $\Pr (c_{\text{R}}<\beta )$ is shown to
decay very fast. Observe that the performance of lattice precoding is
considerably better than that of SVD precoding.

%%%%

%%%%
\begin{figure}[tbp]
\centering
\includegraphics[scale=0.35]{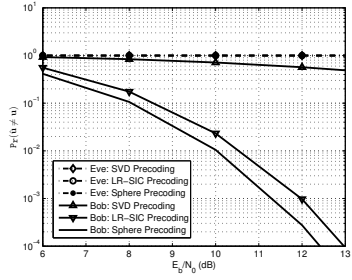}\vspace{-2mm}
\caption{$\Pr \left( \mathbf{\hat{u}\neq u}\right) $ vs. Bob's average SNR
per bit for the uncoded MIMO system with $N_{\text{A}}=10$, $N_{\text{B}}=9$%
, $N_{\text{E}}=20$, 64-QAM and SNR$_{\text{E}}\rightarrow \infty $.
} \label{fig:10}\vspace{-5mm}
\end{figure}
\vspace{-3mm}

\section{Conclusions}

In this paper, we have shown how the artificial noise can force Eve's
received signal to settle around the borders of the decision region, so that
the \emph{practical secrecy} can be achieved. Of particular interest is that
even if only one degree of freedom is used for artificial noise ($N_{\text{A}%
}-N_{\text{B}}=1$) and Eve has unlimited resources ($N_{\text{E}}>N_{\text{A}%
}$), the data can still be protected. The connection between secrecy
capacity and \emph{practical secrecy}, as well as the effect of
finite $N_{\text{B}}$, will be investigated in the future work.
\vspace{-2mm}

\section*{Acknowledgment}

% optional entry into table of contents (if used)
%\addcontentsline{toc}{section}{Acknowledgment}
The authors would like to thank Professor Terence Tao for his constructive
comments and pointing out the references. \vspace{-2mm}

\section*{Appendix}

\subsection{Proof of Theorem 1}

Let $\mathbf{B}=\mathbf{GP}$. We recall the fact \cite{Banaszczyk}:%
\begin{equation}
r_{\text{cov}}(\mathbf{B})\leq \frac{N_{\text{B}}}{\lambda _{1}((\mathbf{B}%
^{\dagger })^{T})}\text{.}  \label{cover_r_upper}
\end{equation}

It is known that for a random lattice basis $\mathbf{B}$ and $N_{\text{B}%
}\rightarrow \infty $, $\lambda _{1}((\mathbf{B}^{\dagger })^{T})$ converges
to \cite{Ajtai02}%
\begin{equation}
\sqrt{N_{\text{B}}/(\pi e)}\left\vert \det ((\mathbf{B}^{\dagger
})^{T})\right\vert ^{1/N_{\text{B}}}\text{.}  \label{lam_gh_app}
\end{equation}%
Consequently, the right hand side of (\ref{cover_r_upper}) tends towards%
\begin{equation}
\frac{\sqrt{\pi eN_{\text{B}}}}{\left\vert \det ((\mathbf{B}^{\dagger
})^{T})\right\vert ^{1/N_{\text{B}}}}=\pi er_{\text{eff}}(\mathbf{B)}\text{.}
\label{CR_UB}
\end{equation}%
If $c_{\text{R}}>\pi e$, using (\ref{Cr_def}), we have
$||\mathbf{GZv||>}r_{\text{cov}}(\mathbf{B})$,
which means $\mathbf{y}\notin \mathcal{V}\left( \Lambda
_{\mathbb{C}}\right) $, so that $P_{\text{E}}=1$. \QEDA

\vspace{-2mm}

\subsection{Proof of Theorem 2}

\emph{1) Lattice precoding:} Performing the QR decomposition yields $\mathbf{%
H}^{T}=\mathbf{Q}_{\text{H}}\mathbf{R}_{\text{H}}$, where $\mathbf{Q}_{\text{%
H}}$ has orthogonal columns and $\mathbf{R}_{\text{H}}$ is an upper
triangular matrix with nonnegative diagonal elements. We have%
\begin{equation}
c_{\text{R}}=\frac{||\mathbf{GZv}||\det (\mathbf{R}_{\text{H}}\mathbf{)}%
^{1/N_{\text{B}}}}{\sqrt{\frac{N_{\text{B}}}{\pi e}}\left\vert \det (\mathbf{%
GQ}_{\text{H}}\mathbf{)}\right\vert ^{1/N_{\text{B}}}}.  \label{r2}
\end{equation}%
We recall the facts that $\mathbf{Q}_{\text{H}}$ is independent of $\mathbf{R%
}_{\text{H}}$ \cite{Verdu04}, and $\mathbf{GZ}$ and $\mathbf{GQ}_{\text{H}}$
are mutually independent Gaussian random matrices \cite{Lukacs54}.
Consequently, $||\mathbf{GZv}||$, $\det (\mathbf{GQ}_{\text{H}}\mathbf{)}%
^{1/N_{\text{B}}}$ and $\det (\mathbf{R}_{\text{H}}\mathbf{)}^{1/N_{\text{B}%
}}$ are mutually independent random variable. It is easy to verify that $%
\frac{\sqrt{2}}{||\mathbf{v}||}||\mathbf{GZv}||$ is a $\mathcal{X}$
distributed random variable with $2N_{\text{E}}$ degrees of freedom, i.e.,%
\begin{equation}
\frac{\sqrt{2}}{||\mathbf{v}||}||\mathbf{GZv}||\rightarrow \mathcal{X(}2N_{%
\text{E}}).  \label{1}
\end{equation}

According to \cite{Tao12}\ , for $N_{\text{B}}\rightarrow \infty $, we have%
\begin{equation}
\frac{\log |\det (\mathbf{R}_{\text{H}})|-1/2\log \frac{N_{\text{A}}!}{(N_{%
\text{A}}-N_{\text{B}})!}+1/4\log (N_{\text{B}})}{1/2\sqrt{\log (N_{\text{B}%
})}}\rightarrow \mathcal{N}(0,1).  \label{dist_gs11}
\end{equation}%
Multiplying the numerator and denominator in (\ref{dist_gs11}) by $1/N_{%
\text{B}}$, we obtain%
\begin{equation}
\det (\mathbf{R}_{\text{H}}\mathbf{)}^{1/N_{\text{B}}}\rightarrow e^{%
\mathcal{N}(0,\frac{1}{4}N_{\text{B}}^{-2}\log (N_{\text{B}}))}\left[ \frac{%
N_{\text{A}}!}{N_{\text{B}}^{1/2}(N_{\text{A}}-N_{\text{B}})!}\right] ^{%
\frac{1}{2N_{\text{B}}}}.  \label{2}
\end{equation}%
Similarly, since $\mathbf{GQ}_{\text{H}}$ is a Gaussian random matrix, for $%
N_{\text{B}}\rightarrow \infty $, we have%
\begin{equation}
\det (\mathbf{GQ}_{\text{H}}\mathbf{)}^{1/N_{\text{B}}}\rightarrow e^{%
\mathcal{N}(0,\frac{1}{4}N_{\text{B}}^{-2}\log (N_{\text{B}}))}\left[ \frac{%
N_{\text{E}}!}{N_{\text{B}}^{1/2}(N_{\text{E}}-N_{\text{B}})!}\right] ^{%
\frac{1}{2N_{\text{B}}}}.  \label{3}
\end{equation}

According to (\ref{1}), (\ref{2}) and (\ref{3}), the right hand side of (\ref%
{r2}) converges to the product of two random variables $f_{\mathcal{X}}\cdot
f_{\mathcal{N}}$, where%
\begin{eqnarray}
f_{\mathcal{X}} &=&\frac{||\mathbf{v}||}{\sqrt{2N_{\text{B}}}}\Phi _{\text{LP%
}}\mathcal{X(}2N_{\text{E}})\text{,}  \notag \\
f_{\mathcal{N}} &=&\sqrt{\pi e}\exp (\mathcal{N}(0,\frac{1}{2}N_{\text{B}%
}^{-2}\log (N_{\text{B}})))\text{,} %\\
%\Phi _{\text{LP}} &=&\left[ \frac{(N_{\text{E}}-N_{\text{B}})!}{(N_{\text{A}%
%}-N_{\text{B}})!}\cdot \frac{N_{\text{A}}!}{N_{\text{E}}!}\right] ^{1/(2N_{%
%\text{B}})}\text{.}  \notag
\end{eqnarray}
with $\Phi _{\text{LP}}$ given in (\ref{Pi_LP}).
Now we compute the probability of $f_{\mathcal{X}}\cdot f_{\mathcal{N}}\leq
\beta $. We have%
\begin{equation}
\Pr \left\{ f_{\mathcal{X}}\cdot f_{\mathcal{N}}\leq \beta \right\} \leq \Pr
\{f_{\mathcal{X}}\leq \beta \}+\Pr \{f_{\mathcal{N}}\leq 1\}.
\end{equation}

We first compute the probability of $f_{\mathcal{N}}\leq 1$:
\[
\Pr \{f_{\mathcal{N}} \leq 1\}=\Pr \left\{ \mathcal{N}(0,\frac{1}{2}N_{%
\text{B}}^{-2}\log (N_{\text{B}}))\leq -\frac{1}{2}\log \pi e\right\}
\]
\begin{equation}
 \leq 1/2\exp \left( -\frac{N_{\text{B}}^{2}\log ^{2}\pi e}{4\log (N_{\text{%
B}})}\right)  \leq O\left( e^{-N_{\text{B}}^{2}/\log (N_{\text{B}})}\right) \text{.}
\end{equation}
Then, we compute the probability of $f_{\mathcal{X}}\leq \beta $:%
\begin{equation}
\Pr \{f_{\mathcal{X}}\leq \beta \}=\Pr \left\{ \mathcal{X}^{2}\mathcal{(}2N_{%
\text{E}})\leq \frac{2\beta ^{2}N_{\text{B}}}{||\mathbf{v}||^{2}\Phi _{\text{%
LP}}^{2}}\right\} \text{.}  \label{b2}
\end{equation}%
Let $||\mathbf{v}||=\frac{\beta e}{\Phi _{\text{LP}}}$. Since $\frac{2\beta
^{2}N_{\text{B}}}{||\mathbf{v}||^{2}\Phi _{\text{LP}}^{2}}=\frac{2N_{\text{B}%
}}{e^{2}}<2N_{\text{E}}$, we have%
\begin{equation}
\Pr \left\{ \mathcal{X}^{2}\mathcal{(}2N_{\text{E}})\leq \frac{2N_{\text{B}}%
}{e^{2}}\right\} \leq (\gamma e^{1-\gamma })^{N_{\text{E}}},
\end{equation}%
where $\gamma =\frac{N_{\text{B}}}{e^{2}N_{\text{E}}}$. It is easy to show
that%
\begin{equation}
\Pr \left\{ \mathcal{X}^{2}\mathcal{(}2N_{\text{E}})\leq \frac{2N_{\text{B}}%
}{e^{2}}\right\} \leq \left[ \frac{e^{2}N_{\text{E}}}{e^{1-\gamma }N_{\text{B%
}}}\right] ^{-N_{\text{E}}}\leq O(e^{-N_{\text{E}}}).
\end{equation}

Therefore, with $||\mathbf{v}||=\frac{\beta e}{\Phi _{\text{LP}}}$%
\begin{eqnarray}
\Pr \{c_{\text{R}} &<&\beta \}=\Pr \{f_{\mathcal{X}}\cdot f_{\mathcal{N}%
}<\beta \}  \notag \\
&\leq &O\left( e^{-\min (N_{\text{B}}^{2}/\log (N_{\text{B}})\text{, }N_{%
\text{E}})}\right) .
\end{eqnarray}

\emph{2) SVD precoding: }The proof is similar to the above. \QEDA

\vspace{-2mm}
\bibliographystyle{IEEEtran}

\end{document}